\documentclass[12 pt, psamsfonts]{amsart}
\usepackage{amsmath}
\usepackage{amsthm}
\usepackage{amsfonts}
\usepackage{amssymb}
\usepackage{eucal}
\usepackage{graphicx}
\usepackage[all,knot]{xy}
\xyoption{arc}

\newcommand{\U}{{\rm U}}
\newcommand{\GL}{{\rm GL}}

\newcommand{\ot}{\otimes}
\newcommand{\B}{\mathcal{B}}

\newcommand{\om}{\omega}
\newcommand{\lan}{\langle}
\newcommand{\ra}{\rangle}

\newcommand{\one}{{\rm Id}}

\newcommand{\Z}{\mathbb{Z}}

\newcommand{\N}{\mathbb{N}}

\newcommand{\C}{\mathbb{C}}
\newcommand{\R}{\mathbb{R}}

\newcommand{\al}{\alpha}
\newcommand{\alp}{\alpha^\prime}

\newcommand{\ii}{{\mathrm i}}

\newtheorem{theorem}{Theorem}[section]
\newtheorem{lemma}[theorem]{Lemma}
\newtheorem{prop}[theorem]{Proposition}
\theoremstyle{definition}

\newtheorem{example}[theorem]{Example}

\newtheorem{question}[theorem]{Question}
\newtheorem{conj}[theorem]{Conjecture}
\newtheorem{cor}[theorem]{Corollary}
\theoremstyle{remark}
\newtheorem{remark}[theorem]{Remark}

\numberwithin{equation}{section}

\begin{document}
\title[Gaussian $(z,N)$-generalized Yang-Baxter operators]{Parameter-dependent Gaussian $(z,N)$-generalized Yang-Baxter operators}
\author{Eric C. Rowell}
\address{Department of Mathematics\\ Texas A\&M University\\ College Station,
Texas}
\email{rowell@math.tamu.edu}

\thanks{E.C.R. is partially supported by US NSF grant DMS-1108725, and wishes to thank V.F.R. Jones, M.-L. Ge and Z. Wang
for illuminating discussions.  Part of this was written while E.C.R. was visiting BICMR, Peking University and the Chern Institute, Nankai University
and he gratefully acknowledges the hospitality of these Institutions.}
\begin{abstract}We find unitary solutions $\tilde{R}(a)$ to the (multipicative parameter-dependent) $(z,N)$-generalized Yang-Baxter equation that carry the standard measurement basis to $m$-level $N$-partite states that generalize the Bell states corresponding to $\tilde{R}(0)$ in the case $m=N=2$.  This is achieved by a careful study of solutions to the Yang-Baxter equation discovered by Fateev and Zamolodchikov in 1982. \end{abstract}

\maketitle

\section{Introduction}
The four Bell states $|\Phi^{\pm}\ra=\frac{1}{\sqrt{2}}(|00\ra\pm|11\ra)$ and 
$|\Psi^{\pm}\ra=\frac{1}{\sqrt{2}}(|01\ra\pm|10\ra)$ are ubiquitous in quantum information: they are maximally entangled bipartite qubit states that play a starring role in quantum teleportation (i.e. the EPR paradox).  Bell states have been generalized to $m$-level bipartite states \cite{WGC} as well as $2$-level $N$-partite states (starting with $N=3$, see \cite{ghz1}).

The Bell basis change matrix $$B:=\frac{1}{\sqrt{2}}\begin{pmatrix}1 &0&0&1\\0&1&1&0\\0&-1&1&0\\-1&0&0&1\end{pmatrix}$$  describes the relationship between the standard qubit measurement basis and the Bell state basis.  Kauffman and Lomonaco \cite{KL} observed that $B$ satisfies the Yang-Baxter equation.  A natural question is:

\begin{question}
 Can find Yang-Baxter operators that produce $m$-level $N$-partite Bell-like states from the measurement basis?
\end{question}

In \cite{GHZ} the generalized Yang-Baxter equation was introduced and solutions associated with extra-special $2$-groups and GHZ states were explored, answering the question above for $m=2$ and all $N$.  In \cite{GHR} this notion was formalized slightly with a discussion in terms of locality.  We say $R\in\GL_{s^z}(\C)$ is a solution to the \textbf{$(z,N)$-generalized Yang-Baxter equation} ($(z,N)$-gYBE) if 
\begin{equation}\label{gybe}
 (R\ot \one_{s^z})(\one_{s^z}\ot R)(R\ot \one_{s^z})=(\one_{s^z}\ot R)(R\ot \one_{s^z})(\one_{s^z}\ot R)
\end{equation}
where $\one_{s^z}$ is the identity operator on $(\C^s)^{\ot z}$.    A \textbf{$(z,N)$-generalized Yang-Baxter operator} is a solution $R\in\GL_{s^N}(\C)$ to the $(z,N)$-gYBE that also satisfied far-commutivity: 
\begin{equation}\label{farcomgyb} (R\ot \one_{s^z}^{\ot j})(\one_{s^z}^{\ot j}\ot R)=(\one_{s^z}^{\ot j}\ot R)(R\ot \one_{s^z}^{\ot j}) \quad j\geq 2
\end{equation}  
When $z=1$ and $N=2$ we recover the ordinary definition of Yang-Baxter operator, and (\ref{farcomgyb}) is automatically satisfied.  Indeed, this is true whenever $N\leq 2z$.  In the same way that the Bell basis change matrix produces the Bell states, the $(z,N)$-gYB operators in \cite{GHZ} produce $N$-partite GHZ-states.  Moreover, these give rise to representations of the braid group, which plays a central role in the topological model for quantum computation (\cite{FKLW}).
A cascade of papers \cite{GHR,Rebecca,Hong1,Hong2,KW} followed these definitions, yielding new solutions and exploring new applications.

A second way to generalize the Bell basis change matrix is to look for $m^2\times m^2$ braiding matrices that produce $m$-level Bell states, e.g. $\frac{1}{\sqrt{m}}\sum_{j=0}^{m-1}c_j|jj\rangle$.  These \emph{Gaussian} solutions to the Yang-Baxter equation were introduced 25 years ago (at least for $m$ an odd prime) in \cite{jonespjm,GJ}.  In explicit matrix form (see \cite{GR}), these are:
$$R=\frac{1}{\sqrt{m}}\sum_{j=0}^{m-1} \om^{j^2} U^j$$ where $\om$ is either an $m$th or $2m$th root of unity (depending on if $m$ is odd or even, respectively) and $U\in \GL(\C^{m^2})$ is defined by $U(|i\rangle\ot|j\rangle)=\om^{i-j}|i-1\rangle\ot|{j-1}\rangle$ where $\{|i\rangle:0\leq i\leq m-1\}$ is the standard basis for $\C^{m}$. The case $m=2$ is equivalent to the Bell basis change matrix.  Recently, Gaussian Yang-Baxter operators 
have experienced something of a renaissance for their connections to quantum information: they describe particle exchange statistics for \emph{metaplectic anyons} \cite{HNW1,HNW2,CW}.  Metaplectic anyons are modeled by the modular categories $SO(N)_2$, as was shown in \cite{RWe}.

The main goal of this article is to extend the results of \cite{GHZ} to all $m>2$, using the Gaussian Yang-Baxter operators.  To do so there are two critical ingredients: 
\begin{enumerate}
\item $(z,N)$-generalized Yang-Baxter operators with the $z=1$, $N=2$ case giving the Gaussian solutions for all $m$, and the $m=2$ case corresponding to the solutions of \cite{GHZ}, and
\item Baxterized (parameter-dependent) versions of these $(z,N)$-gYB operators on $m$-level arrays.
\end{enumerate}

Historically, solutions to the parameter-dependent Yang-Baxter (or star-triangle) equation came before the parameter-independent $R$-matrix solutions that give rise to braid group representations.  Jones \cite{Bax} discussed the reverse process of (Yang-)\emph{Baxterization}: from an $R$-matrix one introduces a spectral parameter, a process which typically depends on studying the spectrum of $R$ itself.  This was explored in the case $R$ has few eigenvalues in \cite{GWX}, which was employed in \cite{GHZ}.  This allowed an explicit description of the Schr\"odinger equation that controls the unitary evolution of the entangled states.

\begin{remark}
 Seven months after an earlier version of this paper was circulated, the paper \cite{WSWLZX} appeared on the arxiv, which has some overlap with our main results, but with a different approach.  The main differences between the two papers are:
 \begin{enumerate}
  \item We give parameter-dependent solutions to the generalized Yang-Baxter equation that give the Gaussian solutions in the limit.  The possibility of Yang-Baxterization is suggested in \cite{WSWLZX}, but the number of eigenvalues of the braiding matrix grows with $m$, so an explicit Yang-Baxterization would be difficult.  The 
  \item The $(z,N)$-generalized Yang-Baxter operators in \cite{WSWLZX} do not all give rise to braid group representations--they do not derive conditions on $z,N$ that gaurantee all braid relations are satisfied.
 \end{enumerate}

\end{remark}

\section{Gaussian YB operators (with spectral parameters)}

For a parameter $q$ the quantum torus $T_{q^2}(n)$ is defined (see \cite{RWe}) to be the
algebra with invertible generators $u_1,\ldots,u_{n-1}$ satisfying: 
\begin{eqnarray}
 &&u_iu_j=u_ju_i\quad |i-j|\neq 1\label{farcomqt}\\
 &&u_iu_{i+1}=q^2u_{i+1}u_i\quad 1\leq i\leq n-2
\end{eqnarray}
Specializing $q\in\C^*$,
 $T_{q^2}(n)$ may be given a $C*$-structure by setting $u_i^*=u_i^{-1}$.  For $q^2$ a primitive $m$th root of unity on sees that $u_i^{m}$ is in the center of $T_{q^2}(n)$, and we denote by $T_{q^2}^m(n)$ the quotient by the relations $u_i^m=1$.

In \cite{FZ}, Fateev and Zamolodchikov define, for any $m\in\N$ the quantities:
$$x_j(\al):=\prod_{k=0}^{j-1}\frac{\sin(\frac{2k\pi+\al}{2m})}{\sin(\frac{2(k+1)\pi-\al}{2m})}.$$
Clearly some $x_j(\alpha)$ are undefined for certain values of $\alpha$, but these will be explored after a change of variables (see below).

For $q^2$ a primitive $m$th root of unity, \cite{FZ} shows that

$$R^{FZ}_i(\al):=\sum_{j=0}^{m-1} x_n(\al)u_i^j$$ 
satisfies the (additive) parameter-dependent Yang-Baxter equation (\emph{star-triangle relation} in \cite{FZ}):
\begin{equation}\label{ybxy}
R_i(\al)R_{i+1}(\al+\alp)R_i(\alp)=R_{i+1}(\alp)R_{i}(\al+\alp)R_{i+1}(\al).\end{equation}
This is achieved by verifying:

\begin{equation}\label{ybcons}
\begin{split}
 \sum_{\ell=0}^{m-1}&x_{n_1-\ell}(\al)x_{n_2}(\al+\alp)x_{n_3-\ell}(\alp)(q^2)^{-n_3\ell}\\&=\sum_{\ell=0}^{m-1}x_{\ell}(\alp)x_{n_1-n_3}(\al+\alp)x_{\ell-n_2}(\al)(q^2)^{-\ell(n_1-n_3)-n_1n_3}
 \end{split}
\end{equation}

for such $q^2$.  In fact, there is a small typo in \cite[eqn. (10)]{FZ}: in their version of eqn. (\ref{ybcons}) the right-hand side has $\al$ and $\alp$ interchanged.

It is immediate from (\ref{farcomqt}) that:
\begin{equation}\label{farcom}
 R_i^{FZ}(\al)R_j^{FZ}(\alp)=R_i^{FZ}(\alp)R_j^{FZ}(\al)\quad |i-j|\neq 1.
\end{equation}

From the considerations in \cite{FZ}, we have the following parameter-dependent analogue of Proposition 3.6(a)(b) from \cite{GHZ}:
 \begin{prop}\label{L:qtreps} Fix $m\in\N$, and suppose that $q^2$ is a primitive $m$th root of unity. Suppose $T_1,\ldots, T_{n-1}\in\GL(V)$ are operators on $V$ satisfying:
 \begin{enumerate}
 \item[(E1)] $T_i^m=\one_V$
 \item[(E2)] $T_iT_j=T_jT_i$ for $|i-j|\neq 1$
 \item[(E3)] $T_iT_{i+1}=q^2T_{i+1}T_i$.
 \end{enumerate}
Then \begin{enumerate}
\item[(a)] The mapping $\phi: T_{q^2}^m(n)\rightarrow \GL(V)$ via $\phi(u_i)=T_i$ extends to a representation of $T_{q^2}^m(n)$.
\item[(b)] $R_i^\phi(\al):=\sum_{j=0}^{m-1}x_j(\al)T_i^j$ satisfies
eqns. (\ref{ybxy}) and (\ref{farcom}). \end{enumerate}
 \end{prop}

 We also wish to address the issue of unitarity.
 For our purposes the  multiplicative parameter-dependent version of (\ref{ybxy}) has some advantages:
\begin{equation}\label{ybxy2}
R_i(a)R_{i+1}(ab)R_i(b)=R_{i+1}(b)R_{i}(ab)R_{i+1}(a).\end{equation}

We reparameterize and rescale $R^{FZ}_i(\alpha)$ as follows:
Set $\al=m\ii\log(1/a)$ (where $\ii=\sqrt{-1}$) and $Q=e^{\pi\ii/m}$ so that:

\begin{equation*} \frac{\sin(\frac{2k\pi+\al}{2m})}{\sin(\frac{2(k+1)\pi-\al}{2m})}=\frac{aQ^k-Q^{-k}}{Q^{k+1}-aQ^{-k-1}}                                                                                                         
\end{equation*}
We then set $$X_j(a):=x_j(m\ii\log(1/a))=\prod_{k=0}^{j-1}\frac{aQ^k-Q^{-k}}{Q^{k+1}-aQ^{-k-1}}.$$  By inspection on sees that the only real singularity occurs at $a=-1$ for $X_{\frac{m}{2}}(a)$ (with $m$ even): the remaining possible singularities for $X_j(a)$ occur at non-real roots of unity $a=Q^{2t}$.  We renormalize  $R_i^{FZ}(m\ii\log(1/a))$ to obtain: $$\tilde{R}_i(a):=\sum_{j=0}^{m-1}\left(\frac{(a+1)(a^m-1)}{m(a-1)(a^m+1)}\right)^{\frac{1}{2}}X_j(a)u_i^j.$$  Setting $\tilde{X}_j(a)=\left(\frac{(a+1)(a^m-1)}{m(a-1)(a^m+1)}\right)^{\frac{1}{2}}X_j(a)$ we note that these quantities are well-defined for all real numbers $a$.  Indeed the order 1 pole of $X_{\frac{m}{2}}(a)$ at $a=-1$ cancels the order 1 zero of $\sqrt{(a+1)(a^m-1)}$ at $a=-1$: we obtain $\tilde{R}_i(-1)=\ii(-u_i)^\frac{m}{2}$ by calculating the limit.

Note that $\tilde{X}_j(a)$ converges in the limits $a\rightarrow \pm\infty$.  On the other hand, $\tilde{X}_j(1)=0$ for $1\leq j\leq m-1$ and $\tilde{X}_0(1)=1$ so that we recover the trivial $R_i=I$ solution.

We can now prove a parameter-dependent version of Proposition 3.6(c) of \cite{GHZ}:
\begin{prop}\label{L:unitary} Keep the hypotheses of Proposition \ref{L:qtreps} and assume that in addition
 the $T_i^{\dag}=T_i^{-1}$ (so $T_i$ are all unitary) and $a\in\R$.  Then:
\begin{enumerate}
\item[(a)] $\phi:T_{q^2}^m(n)\rightarrow U(V)$ and
\item[(b)]  $\tilde{R}_i^\phi(a):=\sum_{j=0}^{m-1}\tilde{X}_j(a)T_i^j\in U(V)$
\end{enumerate}
where $Q=e^{\pi\ii/m}$ as above.
\end{prop}
\begin{proof} Part (a) has already been proved in \cite{GR}.  From the calculation $\tilde{R}_i(-1)=\ii(-u_i)^\frac{m}{2}$ above, we have $\tilde{R}_i^{\phi}(-1)=\ii(-T_i)^{\frac{m}{2}}$ for $m$ even, which is clearly unitary.

Thus we may assume that either $m$ is odd or $a\neq -1$.  We will work with the un-normalized coefficients $X_j(a)$ and derive the normalization factor $\left(\frac{(a+1)(a^m-1)}{m(a-1)(a^m+1)}\right)^{\frac{1}{2}}$.
 For real values of $a$, we have $\overline{X_j(a)}=X_j(1/a)$ so that (b) follow once we establish:
 
 \begin{equation}\label{eq:unit}
  \sum_{n=0}^{m-1}X_n(a)X_{n+j}(1/a)=\begin{cases} 0 & 0<j\leq m-1\\ \frac{m(a-1)(a^m+1)}{(a+1)(a^m-1)} & j=0.\end{cases}
 \end{equation}
 For $j=0$, we compute $$\sum_{n=0}^{m-1}X_n(a)X_n(1/a)=(a-1)^2\sum_{n=0}^{m-1}\frac{Q^{2n}}{(a-Q^{2n})(aQ^{2n}-1)}.$$  Setting $r=Q^2$ we obtain: $$(a-1)^2\sum_{n=0}^{m-1}\frac{1}{(a-q^n)(a-q^{-n})}=\frac{m(a-1)(a^m+1)}{(a+1)(a^m-1)}$$ giving the claimed normalization factor.  Notice that when $m$ is odd this quantity does not vanish at $a=-1$ since $a^m+1$ appears in the numerator.  It remains to verify (\ref{eq:unit}) for $j>0$, which we compute:
 $$\left( a-1 \right) ^{2}{Q}^{j}\sum _{n=0}^{m-1}  {\frac {{Q}^{2n}
\prod _{i=0}^{j-2}({Q}^{2n+2+2i}-a)}{\prod _{i=0}^{j}(a{Q}^{2n+2i
}-1)}}.$$  Setting $r=Q^2$ as above and removing the factors $(a-1)^2$ and $Q^j$ we obtain:
$$\sum _{n=0}^{m-1} {r}^{n} {\frac { \prod _{i=0
}^{j-2}({r}^{n+1+i}-a)}{\prod _{i=0}^{j}(a{r}^{n+i}-1)}} 
.$$  As in \cite{overflow} we use the fact that $\prod_{i=0}^{m-1}(a-r^k)=(a^m-1)$ to rewrite the summands: $${r}^{n} {\frac { \prod _{i=0
}^{j-2}({r}^{n+1+i}-a)}{\prod _{i=0}^{j}(a{r}^{n+i}-1)}}=\frac{C}{(a^m-1)}r^{-nj}\prod_{i=1}^{j-1}(a-r^{n}r^i)\prod_{i=1}^{m-j-1}(a-r^{-n}r^i),$$
where $C=(-1)^{j+1}r^{-j(j+1)/2}$.  Setting $t=r^{-n}$ the summands are (up to an overall constant) $P(t):=t^j\prod_{i=1}^{j-1}(a-t^{-1}r^i)\prod_{i=1}^{m-j-1}(a-tr^i)$.  We must show that for each $j$, $\sum_{s=0}^{m-1}P(r^{s})=0$.  For this, notice that $P(t)$ is a polynomial in $t$, and each monomial has degree strictly between $1$ and $m-1$.  Thus each coefficient of $a$ in $\sum_{s=0}^{m-1}P(r^{s})$ has a factor of the form $\sum_{n=0}^{m-1}r^{nk}$ where $1\leq k\leq m-1$, which vanishes.
\end{proof}

\begin{example}
  Let us pause to compare this to \cite{GHZ}, i.e. the case $m=2$.  In that paper, Proposition 3.6, relation (E1) is replaced by $T_k^2=-\one_V$ and the condition for unitary is that $T_k^{\dag}=-T_k$, i.e. the $T_k$ are all anti-Hermitian.  If we rescale $T_k$ by $\ii$ then our conditions match.  Moreover, we have $$\tilde{R}_k(a)=\tilde{X}_0(a)\one_V+\tilde{X}_1(a)T_k=\frac{1}{\sqrt{2a^2+2}}[(a+1)\one_V+(1-a)(\ii T_k)],$$ which matches the form of the unitary Yang-Baxterized solution of \cite[eqn. (4.23)]{GHZ} after rescaling $T_k$ by $\ii$.
\end{example}

 \section{parameter-dependent $(z,N)$-generalized Yang-Baxter operators}

 With Propositions \ref{L:qtreps} and \ref{L:unitary} in hand, we may 
mimic the approach of \cite[Theorem 3.21]{GHZ} to obtain local $m^{N+z(n-2)}$-dimensional representations of $T_{q^2}^m(n)$, where $q^2$ is a primitive $m$th root of unity.  That is, we construct unitary matrices $M\in\U(m^N)$ so that $$T_i=\one_{m^z}^{\ot (i-1)}\ot M\ot\one_{m^z}^{\ot (n-i-N+1)}$$ satisfying Proposition \ref{L:qtreps}.

Define generalized Pauli operators on $\C^m$ with basis $[|0\rangle,\ldots,|{m-1}\rangle]$ as 
 $\sigma_x(|i\rangle)=q^{i}|{i-1}\rangle$ and $\sigma_y(|i\rangle)=q^{-i}|{i-1}\rangle$ where $|i\pm m\rangle:=|i\rangle$.   Now define $$M_{m^N}:=q^{\frac{(m-1)(N-2)}{2}}\sigma_x\ot\sigma_y^{\ot N-1}$$ on the vector space $(\C^m)^{\ot N}$.  We have:
 \begin{theorem} The assignment $\psi(u_i)=\one_{m^z}^{\ot i-1}\ot M_{m^N}\ot \one_{m^z}^{\ot n-i-1}$ defines a unitary $m^{N+(z-2)n}$-dimensional representation of $T_{q^2}^m(n)$ provided $\frac{N}{2}\leq z\leq N-1$.
 \end{theorem}
 \begin{proof} Clearly $\sigma_x$ and $\sigma_y$ are themselves unitary so the operators $\psi(u_i)$ are also unitary.  Since $\sigma_x^m=q^{\frac{m(m-1)}{2}}\one_{m}$ and $\sigma_y^m=q^{\frac{-m(m-1)}{2}}\one_m$ we have 
 $$(M_{m^N})^m=q^{\frac{m(m-1)(N-2)}{2}}(\sigma_x\ot\sigma_y^{\ot N-1})^m=\one_{m^N}.$$   Next we compute: $\sigma_x\sigma_y=q^{-2}\sigma_y\sigma_x$, so as long as $z\leq N-1$ we have 
 $\psi(u_i)\psi(u_{i+1})=q^2\psi(u_{i+1})\psi(u_i)$.  Indeed, only the $(z+1)$st tensor factors of $\psi(u_i)$ and $\psi(u_{i+1})$ do not commute--they are $\sigma_y$ and $\sigma_x$ respectively, yielding the factor of $q^{2}$.
 
Similarly, the condition $\psi(u_i)\psi(u_j)=\psi(u_j)\psi(u_i)$ for $|i-j|>1$ holds precisely when $2z\geq N$.  Thus we have verified the conditions of Proposition \ref{L:qtreps}.
 \end{proof}
 
In particular we obtain parameter-dependent solutions to the $(N,z)$-generalized Yang-Baxter equation via:

 $$R_i^\psi(\al):=\sum_{j=0}^{m-1}x_j(\al)\one_{m^z}^{\ot i-1}\ot (M_{m^N})^j\ot \one_{m^z}^{\ot n-i-1}$$
 which also satisfy (\ref{farcom}), and are unitary provided $\al\in\ii\R$.

\subsection{Parameter-free $(N,z)$-generalized Yang-Baxter operators}

  Since $\log(1/(ab))=\log(1/a)+\log(1/b)=\frac{\al+\alp}{m\ii}$ the operators:
$$R_i(a):=\sum_{j=0}^{m-1}X_j(a)u_i^j$$ satisfy (\ref{ybxy2}) with $q^2$ a primitive $m$th root of unity.  Consequently, $$\tilde{R}_i:=\tilde{R}_i(0)=\frac{1}{\sqrt{m}}\sum_{j=0}^{m-1}Q^{(mj-j^2)}u_i^j$$
 gives a representation of $\B_n$ into $T_{q^2}^m(n)$ via $\sigma_i\rightarrow \tilde{R}_i$ for any choice of $q^2$ a primitive $m$th root of unity.
Notice that $e^{(mj-j^2)\pi\ii/m}=e^{(m-1)\pi\ii j^2/m}$.  Now $e^{2\pi\ii(m-1)/m}=e^{-2\pi\ii/m}$ is a primitive $m$th root of unity for any $m$, so we may choose $q=e^{\pi\ii (m-1)/m}=-e^{-\pi\ii/m}$ and obtain a representation $\Xi:\B_n\rightarrow T_{q^2}^m(n)$ via $$\Xi(\sigma_i)= S_i:=\frac{1}{\sqrt{m}}\sum_{j=0}^{m-1}q^{j^2}u_i^j.$$  When $m$ is odd, $q$ is a primitive $m$th root of unity, whereas when $m$ is even, $q$ is Galois conjugate to $e^{\pi\ii/m}$ so in either case we may apply a Galois automorphism to recover the Gaussian representation of \cite[Proposition 3.1]{GR}. 
Applying the Proposition \ref{L:unitary} we obtain unitary representations of $\B_n$ via
 $$S_i^\psi:=\frac{1}{\sqrt{m}}\sum_{j=0}^{m-1}q^{j^2}\one_{m^z}^{\ot i-1}\ot (M_{m^N})^j\ot \one_{m^z}^{\ot n-i-1}$$

\section{Conclusions and Discussion}

 The operator $S^\psi=\frac{1}{\sqrt{m}}\sum_{j=0}^{m-1}q^{j^2} (M_{m^N})^j$ carries the standard basis for $(\C^m)^{\ot N}$ to a basis of entangled states.  For a concrete example suppose that $m$ is odd and $q$ is a primitive $m$th root of unity (so that $q^2$ is also a primitive $m$th root of unity).  Then
$$S^\psi |k\rangle^{\ot N}=\frac{1}{\sqrt{m}}\sum_{j=0}^{m-1}q^{c_j(k,m,N)}|j\rangle^{\ot N}$$
where $c_j(k,m,N):=(k-j)^2+\frac{[m-1+(j-k)(j+k+1)](N-2)}{2}$.
For $N=2$ corresponding to the standard Yang-Baxter equation one obtains Gaussian coefficients $c_j(k,m,2)=q^{(k-j)^2}$.  It is clear that one obtains $N$-partite $m$-level generalizations of the Bell states from the other states in the measurement basis.

Regarding $a$ as the time variable, we can consider the unitary evolution of an initial state $\varphi(0)$ via $\tilde{R}(a)\varphi(0)=\varphi(a)$.  Notice that on the interval $0\leq a\leq 1$ the function $\tilde{R}(a)$ interpolates between the Gaussian solution and the trivial solution $Id$, with Gaussian solutions at $\pm\infty$ as well.  (Of course, the only values of $a$ where $\tilde{R}(a)$ satisfies the parameter-free multiplicative Yang-Baxter equation are $a\in\{0,1,\pm\infty\}$ so these are the only values for which we obtain representations of the braid group $\B_n$.)  The Schr\"odinger equation governing this unitary evolution is discussed at length in \cite[Section 4.3]{GHZ}, from which one may derive the time-dependent Hamiltonian.

\thispagestyle{empty}

\end{document}